\documentclass[11pt]{article}
\usepackage{enumitem}
 \usepackage{amsfonts}
\usepackage{amsmath}
\usepackage{graphicx}
\usepackage{fancyhdr}
\usepackage{amsthm}
\usepackage{txfonts}
\usepackage{amssymb}
\usepackage{amscd}
\usepackage{float}
\usepackage{color}
\usepackage{authblk}
\newtheorem{definition}{Definition}[section]
\newtheorem{theorem}{Theorem}[section]

\newtheorem{lemma}{Lemma}[section]

\newtheorem{remark}{Remark}[section]

\begin{document}

\title{{Mean-variance-utility portfolio selection with time and state dependent risk aversion}}
\author{{Ben-Zhang Yang$^a$, Xin-Jiang He$^b$, and Song-Ping Zhu$^b$\thanks{Corresponding author. E-mail address: spz@uow.edu.au}}\\
{\small\it a. School of Mathematics, Sichuan University, Chengdu 610064, P.R. China}\\
{\small\it b. School of Mathematics and Applied Statistics, University of Wollongong, NSW 2522, Australia}
}
\date{}
\maketitle
\vspace*{-9mm}
\begin{center}
\begin{minipage}{5.5in}
{\bf Abstract.} Under mean-variance-utility framework, we propose a new portfolio selection model, which allows wealth and time both have influences on risk aversion in the process of investment. We solved the model under a game theoretic framework and analytically derived the equilibrium investment (consumption) policy.  The results conform with the facts that optimal investment strategy heavily depends on the investor's wealth and future income-consumption balance as well as the continuous optimally consumption process is highly dependent on the consumption preference of the investor.
\\ \ \\

\end{minipage}
\end{center}
\section{Introduction}

Since Markowitz's pioneering work on a static portfolio selection model \cite{Markowitz52}, mean-variance problem has become one of the most important tools in finance to achieve a balance between uncertain returns and risks. Under mean-variance framework, there are several models have been proposed and developed to address investment problems, which have attracted a lot of attention from both academic researchers and market practitioners \cite{Ma,Yang19,Yang20a,Yang20b}.

Within the complete market setting, various pre-commitment (or time-inconsistent) results have been presented for the variance-minimizing policy using martingale methods, given that the expected terminal wealth is equal to a certain level (see \cite{Besnainou98,Bielecki05,Cvitanic08,Cvitanic04,Zhao02}). In an incomplete market, Cochrane \cite{Cochrane08} derived the optimal investment policy that minimizes the "long-term" variance of portfolio returns subject to the constraint that the long-term mean of portfolio returns equals to a pre-specified target level. This approach has also been applied in futures trading strategies by Duffie and Richardson \cite{Duffie91} through setting a mean-variance objective at the initial date. They obtained a pre-commitment solution, which also solves the optimal problem with a quadratic objective for some specific parameters. A similar approach developed for continuous time complete-market settings has also been widely discussed in the literature \cite{Brandt09,Leippold04,Li00,Lim02,Zhou00}.

However, Basak and Chabakauri \cite{Basak10} challenged the pre-commitment assumption \cite{Zhou00}, and assumed investors are sophisticated in the sense that they will maximize their mean-variance objective over time considering all future updates, instead of finding an optimal solution at a fixed given time moment. Following this, Kryger and Steffensen \cite{Kryger10} worked under the Black-Scholes framework without the pre-commitment assumption, and showed that the optimal strategy derived for a mean-standard deviation investor is to take no risk at all. Bj\"{o}rk et al. \cite{Bjork12} further considered  mean-variance optimization problems under a game theoretic framework, and the optimal strategies were derived in the context of sub-game perfect Nash equilibrium.

Recently, researchers started to incorporate consumption choices into the mean-variance problem, investigating the optimal investment-consumption problem together with the mean-variance criterion. For example, Kronborg and Steffensen \cite{Kronborg15} directly added the accumulated consumption to the terminal wealth to formulate an ``adjusted" terminal wealth, and tried to maximize the adjusted terminal wealth over time under the mean-variance framework. Christiansen and Steffensen \cite{Christiansen13} further considered the same optimization problem with deterministic consumption and investment to avoid a series of difficulties. Unfortunately, the optimal consumption strategy derived under this particular model has caused a probably absurd conclusion that investor could suddenly be required to switch his/her optimal consumption strategy from consuming as much as possible to as little as possible. To obtain a rational consumption policy, Yang et. al \cite{Yang20} proposed a new portfolio selection model which simultaneously maximize the terminal wealth and accumulated consumption utility subject to a mean variance criterion controlling the final risk of the portfolio.  The analytically derived policy performs the continuous influence of investors' consumption preference on the optimal consumption strategy and represents a more economically rational investment/consumption behavior.

Unfortunately, the optimal amount of moment to invest in is not dependent of wealth under the model setting \cite{Yang20}, which means that for a given risk aversion degree a rich or poor investor optimally invest the same amount of the money in stocks. In fact, the investor will change his investment policy according to the update of her/his wealth in the case of multi-stage investment.
Inspired this, we propose a new portfolio selection model which allows the risk aversion depend on present wealth and time, therefore the progression of wealth and time phasing can have impacts on the varying risk aversion. The newly formulated optimization problem still preserves the analytical tractability under a continuous-time game theoretic framework, and the analytical optimal continuous investment and consumption strategies derived in the sense of equilibrium \cite{Bjork09,Bjork12} admit intuitive economic explanation.

The rest of this paper is organized as follows. Section 2 proposes the new portfolio selection problem. In Section 3, we analytically derive the optimal strategies based on the definition of the equilibrium strategy. Some concluding remarks are given in the last section.

\section{The portfolio selection problem}

\subsection{The model}

We now assume that we work under the standard Black-Scholes market, where an investor has access to a risk-free bank account and a stock whose dynamics can be specified as
\begin{equation}
\begin{aligned}
dM(t)&=rM(t)dt, &&\quad M(0)=1,\\
dS(t)&=\mu S(t)dt+\sigma S(t) dB(t), && \quad S(0)=s_0>0.
\end{aligned}
\end{equation}
Here, $r>0$, $\mu$ and $\sigma$ are constants, and it is assumed that $\mu>r$. The process $B(t)$ is a standard Brownian motion on the probability space $(\Omega, \mathcal{F},\mathbb{P})$ with the filtration $\sigma\{B(s);0\leq s \leq t\}$, $\forall t \in [0,T]$.

Let $L^2_{\mathcal{F}}(0,T;R)$ denote the set of all $R$-valued, measurable stochastic process $f(t)$ adapted to $\{F_t\}_{t\geq 0}$ such that $E\left[\int_{0}^{T}f^2(t)dt\right]<\infty$.
We also assume that the investor in this market needs to make investment decisions on a finite time horizon $[0,T]$, and he/she allocates a proportion $\pi(t)$ and $1-\pi(t)$ of his wealth into the stock and bank account, respectively, at time $t$. Let $X^{\pi}(t)$ be the wealth of the investor at time $t$ following the investment strategy $\pi(\cdot)$ with an initial wealth of $x_0$ at time $0$. We assume that the investor possesses a continuous deterministic income rate $l(t)$, and chooses a non-negative consumption rate $c(t)$. Under these assumptions, the dynamic of the investor's wealth can be derived as
\begin{equation}\label{wealth}
\left\{
\begin{aligned}
dX^{c,\pi}(t)&=[(r+\pi(t)(\mu-r))X^{c,\pi}(t)+l(t)-c(t)]dt+\pi(t)\sigma X^{c,\pi}(t)dB(t), \quad t \in [0,T),\\
X(0)&=x_0>0.
\end{aligned}
\right.
\end{equation}

In this paper, by introducing a time and sate dependent risk aversion function \cite{Bjork12,Kronborg15}, we propose a general portfolio selection model: at a given time $t$, the investor attempt to achieve the following objective
\begin{equation}\label{maxu2}
\begin{aligned}
&\max_{c(\cdot),\pi(\cdot)} && E\left(e^{-\delta (T-t)}X(T)\right)-\frac{\gamma}{2(x+K^{(c)}(t,x))}Var\left(e^{-\delta (T-t)}X(T)\right)+\beta E\left(\int_{t}^{T}e^{-\rho(s-t)}U(c(s))ds\right)\\
&s.t. && \left\{
\begin{aligned}
&c(\cdot), \pi(\cdot) \in L^2_{\mathcal{F}}(0,T;R),\\
&(X(\cdot),c(\cdot),\pi(\cdot)) ~satisfy ~Equation ~ \eqref{wealth},
\end{aligned}
\right.
\end{aligned}
\end{equation}
where
\begin{equation}
K^{(c)}(t,x):=E\left[\left.\int_{t}^{T}e^{-r(s-t)}\left(l(s)-c(s,X^{c,\pi}(s))\right)ds\right|X(t)=x\right]
\end{equation}
is the time-$t$ financial value of future labor income net of consumption, where $\gamma>0$ is the risk-aversion parameter and $\beta>0$ is the preference coefficient of total utility of consumption.

Obviously, the model \eqref{maxu2} involves the optimization of the mutual objective of expected return, risk and consumption utility. It is necessary for the investor to consider selecting a set of appropriate investment and consumption strategies to achieve the goal of both maximizing return and minimizing risk as well as maximizing consumption utility. It should be highlighted that, except for being relevant to the original constant risk-aversion coefficient $\gamma$, the preference of risk tolerance is also dependent with the time and the investor's  wealth. If the investor's future income and expenditure is relatively excellent, he will reduce the corresponding risk aversion degree, and will tend to invest more in the stock to obtain the potential outcomes.

We would like to point out that  the model \eqref{maxu2} includes several known models as special cases. In fact, if we remove the consumption component, the model degenerates into the one studied in \cite{Christiansen13,Kronborg15}; if we do not consider the time and sate dependent risk aversion function, the model becomes the one reported in \cite{Yang20}; if $\delta$ is set to 0, the time and sate dependent risk aversion and the consumption component are not be taken into account, then the model becomes the classical mean-variance model (see \cite{Basak10,Li00,Zhou00}).

\subsection{Equilibrium strategy}

We shall solve the optimal portfolio selection problem \eqref{maxu2} under a game theoretic framework, which was introduced in \cite{Bjork09,Bjork12} and developed by \cite{Kronborg15,Yang20}. The equilibrium strategy under the continuous-time game theoretic equilibrium for the problem \eqref{maxu2} can be defined as follows.
\begin{definition}\label{def1}
Consider a strategy $(c^*,\pi^*)$ and a fixed point $(c,\pi)$. For a fixed number $h>0$ and an initial point $(t,x)$, we define the strategy $(\widetilde{c}_h,\widetilde{\pi}_h)$ as
\begin{equation}\label{pi0}
(\widetilde{c}_h(s),\widetilde{\pi}_h(s))=\left\{
\begin{aligned}
&(c,\pi), && \text{for} \quad  t\leq s< t+h,\\
&(c^*(s),\pi^*(s)), &&\text{for} \quad t+h \leq s <T.
\end{aligned}
\right.
\end{equation}
If
\begin{equation}
\lim_{h\rightarrow 0} \inf \frac{1}{h}\left(f^{c^*,\pi^*}(t,x,y^{c^*,\pi^*},z^{c^*,\pi^*},w^{c^*,\pi^*})-f^{\widetilde{c}_h,\widetilde{\pi}_h}(t,x,y^{\widetilde{c}_h,\widetilde{\pi}_h},z^{\widetilde{c}_h,\widetilde{\pi}_h},w^{\widetilde{c}_h,\widetilde{\pi}_h})\right)\geq 0
\end{equation}
for all $(c,\pi) \in \mathbb{R}_+ \times \mathbb{R}$, where $f$ is  an optimal value function and
\begin{equation}\label{yzw0}
\begin{aligned}
&y^{c,\pi}:=y^{c,\pi}(t,x)=E\left[\left.e^{-\delta(T-t)}X^{c,\pi}(T)\right|X(t)=x\right],\\
&z^{c,\pi}:=z^{c,\pi}(t,x)=E\left[\left.\left(e^{-\delta(T-t)}X^{c,\pi}(T)\right)^2\right|X(t)=x\right],\\
&w^{c,\pi}:=w^{c,\pi}(t,x)=E\left[\left.\int_{t}^{T}e^{-\rho(s-t)}U(c(s))ds\right|X(t)=x\right],\\
\end{aligned}
\end{equation}
then $(c^*,\pi^*)$ is an equilibrium strategy.
\end{definition}

If we denote $(c^*,\pi^*)$ as the equilibrium strategy satisfying Definition \ref{def1}, and let $V$ be the  the corresponding value function with the equilibrium strategy, we can obtain
\begin{equation}\label{op1}
V(t,x)=f^{c,\pi}(t,x,y^{c^*,\pi^*},z^{c^*,\pi^*},w^{c^*,\pi^*}).
\end{equation}
Clearly, our problem is to search for the corresponding optimal strategies and the optimal value function $f:[0,T]\times \mathbb{R}^4\rightarrow R$ as a $\mathcal{C}^{1,2,2,2,2}$ function of the form
\begin{equation}\label{ff}
f^{c^*,\pi^*}(t,x,y^{c,\pi},z^{c,\pi},w^{c,\pi})=y-\frac{\psi(t,x)}{2}(z-y^2)+\beta w, \quad (c,\pi)\in \mathcal{A},
\end{equation}
where $\psi(t,x)=\frac{\gamma}{x+K^{(c)}(t,x)}$ and $\mathcal{A}$ is the class of admissible strategies.

Before we are able to present the optimal solution, some preliminaries need to be outlined. As reported in studies \cite{Kronborg15,Yang20}, we can  establish an extension of the HJB equation for the characterization of the optimal value function and the corresponding optimal strategy, so that the stochastic problem can be transformed into a system of deterministic differential equations and a deterministic point-wise minimization problem. We introduce the following two lemmas. Due to the length limitation, we are not prepared to prove the following lemmas and recommend interested readers to refer to the literature \cite{Yang20}.
\begin{lemma}
Suppose there exist three functions $Y=Y(t,x)$, $Z=Z(t,x)$ and $W=W(t,x)$ such that
\begin{equation}\label{Y1}
\left\{
\begin{aligned}
Y_t(t,x)&=-[(r+\pi(\mu-r))x+l-c]Y_x(t,x)-\frac{1}{2}\pi^2\sigma^2x^2Y_{xx}(t,x)+\delta Y(t,x),\\
Y(T,x)&=x,
\end{aligned}
\right.
\end{equation}
\begin{equation}\label{Z1}
\left\{
\begin{aligned}
Z_t(t,x)&=-[(r+\pi(\mu-r))x+l-c]Z_x(t,x)-\frac{1}{2}\pi^2\sigma^2x^2Z_{xx}(t,x)+2\delta Z(t,x),\\
Z(T,x)&=x^2,
\end{aligned}
\right.
\end{equation}
and
\begin{equation}\label{W1}
\left\{
\begin{aligned}
W_t(t,x)&=-[(r+\pi(\mu-r))x+l-c]W_x(t,x)-\frac{1}{2}\pi^2\sigma^2x^2W_{xx}(t,x)-e^{-\rho t}U(c),\\
W(T,x)&=0,
\end{aligned}
\right.
\end{equation}
where $(c,\pi)$ is an arbitrary admissible strategy. Then,
\begin{equation}\label{s1}
Y(t,x)=y^{c,\pi}(t,x),\quad Z(t,x)=z^{c,\pi}(t,x), \quad W(t,x)=w^{c,\pi}(t,x),
\end{equation}
where $y^{c,\pi}$, $z^{c,\pi}$  and $w^{c,\pi}$ are given by  \eqref{yzw0}.
\end{lemma}

\begin{lemma}\label{lemma2}
If there exists a function $F=F(t,x)$ such that
\begin{equation}\label{F1}
\left\{
\begin{aligned}
&F_t=\inf_{c,\pi \in \mathcal{A}} \left\{-[(r+\pi(\mu-r))x+l-c](F_x-Q)-\frac{1}{2}\pi^2\sigma^2x^2(F_{xx}-U)+J\right\},\\
&F(T,x)=f^{c,\pi}(T,x,x,x^2,0),
\end{aligned}
\right.
\end{equation}
where $Q=f_x^{c^*,\pi^*}$,
\begin{equation}\label{e30}
\begin{aligned}
U=&f_{xx}^{c^*,\pi^*}+f_{yy}^{c^*,\pi^*}(F^{(1)}_{x})^2++f_{zz}^{c^*,\pi^*}(F^{(2)}_{x})^2+f_{ww}^{c^*,\pi^*}(F^{(3)})^2+2f_{xy}^{c^*,\pi^*}F^{(1)}_x+2f_{xz}^{c^*,\pi^*}F^{(2)}_x\\
&+2f_{xw}^{c^*,\pi^*}F^{(3)}_x+2f_{yz}^{c^*,\pi^*}F^{(1)}_xF^{(2)}_x+2f_{yw}^{c^*,\pi^*}F^{(1)}_xF^{(3)}_x+2f_{zw}^{c^*,\pi^*}F^{(2)}_xF^{(3)}_x
\end{aligned}
\end{equation}
and
\begin{equation}\label{e31}
J=f_t^{c^*,\pi^*}+f_y^{c^*,\pi^*}\delta F^{(1)}+2f_z^{c^*,\pi^*}\delta F^{(2)}-f_w^{c^*,\pi^*}e^{-\rho t}U(c(t)).
\end{equation}
with
$$F^{(1)}=y^{c^*,\pi^*}(t,x), \quad F^{(2)}=z^{c^*,\pi^*}(t,x), \quad F^{(3)}=w^{c^*,\pi^*}(t,x),$$
then
$$F(t,x)=V(t,x),$$
where $V$ is the optimal value function defined by \eqref{op1}.
\end{lemma}

\section{Determination of optimal strategy}

In this section, we present the optimal solutions to the optimal portfolio selection problem \eqref{maxu2} based on the results derived in the previous section, and some detailed discussions are provided to illustrate the behaviour of the optimal strategies.

\begin{lemma}\label{l33}
The optimal policy for the optimal value function \eqref{F1} can be solved as
\begin{equation}\label{pi1}
\pi^*=-\frac{\mu-r}{x\sigma^2}\frac{F^{(1)}_x+\psi F^{(1)}F^{(1)}_x-\frac{\psi}{2}F^{(2)}_x+\beta F^{(3)}_x}{F^{(1)}_{xx}+\psi F^{(1)}F^{(1)}_{xx}-\frac{\psi}{2}F^{(2)}_{xx}+\beta F^{(3)}_{xx}}
\end{equation}
and
\begin{equation}\label{cstar1}
c^*=[U']^{-1}\left(\frac{1}{\beta}e^{-\rho t}R(t,x)\right),
\end{equation}
where $[f]^{-1}(\cdot)$ is the inverse function of $f$ and
\begin{equation}
R(t,x)=F^{(1)}_x+\psi F^{(1)}F^{(1)}_x-\frac{\psi}{2}F^{(2)}_x+\beta F^{(3)}_x+\frac{\gamma}{2(x+K^{(c)})^2}(1+K_x^{c^*}(t,x))\left(F^{(2)}-(F^{(1)})^2\right).
\end{equation}
\end{lemma}
\begin{proof}
A candidate strategy for the optimal value function \eqref{F1} can be derived by simply differentiating \eqref{F1} with respect to $\pi$ and $c$, respectively. Therefore, the optimal strategy $\pi^*$  should satisfy
\begin{equation}
\frac {\partial}{\partial \pi}\left(-\pi(\mu-r)x(F_x-Q)-\frac{1}{2}\pi^2\sigma^2x^2(F_{xx}-U)\right)=0.
\end{equation}
A further simplification then yields
\begin{equation}\label{o1}
\pi^*=-\frac{\mu-r}{x\sigma^2}\frac{F_x-Q}{F_{xx}-U}.
\end{equation}
Recall the corresponding objective form
\begin{equation}\label{ff2}
f(t,x,y,z,w)=y-\frac{\psi(t,x)}{2}(z-y^2)+\beta w,
\end{equation}
where $\psi(t,x)=\frac{\gamma}{x+K^{(c)}(t,x)}$. Substituting \eqref{ff2}
into \eqref{e30} and \eqref{e31} gives
\begin{equation}\label{z}
\begin{aligned}
F_x-Q=&F^{(1)}_x+\psi F^{(1)}F^{(1)}_x-\frac{\psi}{2}F^{(2)}_x+\beta F^{(3)}_x,\\
F_{xx}-U=&F^{(1)}_{xx}+\psi F^{(1)}F^{(1)}_{xx}-\frac{\psi}{2}F^{(2)}_{xx}+\beta F^{(3)}_{xx}.\\
\end{aligned}
\end{equation}
Similarly, we can also obtain
\begin{equation}\label{JAA}
J=\delta F^{(1)}-(\psi\delta+\frac{\psi_t}{2})\left(F^{(2)}-(F^{(1)})^2\right)-\beta e^{-\rho t}U(c).
\end{equation}
By characterizing as the solution to a Feynman-Kac PDE, we can obtain
\begin{equation}\label{psi1}
\psi_t=-\frac{\gamma}{(x+K^{(c)})^2}\left(rK^{(c)}-l+c-(rx+l-c)K^{(c)}_x-\frac{1}{2}\pi^2\sigma^2x^2K^{(c)}_{xx}\right).
\end{equation}
Inserting \eqref{psi1} into \eqref{JAA}, we have the new form of $J$ as follows
\begin{equation}\label{JAB}
\begin{aligned}
J=&\delta F^{(1)}-\frac{\gamma \delta}{x+K^{(c)}}\left(F^{(2)}-(F^{(1)})^2\right)
+\frac{\gamma}{2(x+K^{(c)})^2}\bigg(rK^{(c)}-l+c-(rx+l-c)K^{(c)}_x-\frac{1}{2}\pi^2\sigma^2x^2K^{(c)}_{xx}\bigg)\\
&
\times\left(F^{(2)}-(F^{(1)})^2\right)-\beta e^{-\rho t}U(c).\\
\end{aligned}
\end{equation}
By substituting \eqref{JAB} into optimal value function \eqref{F1} and differentiating with respect to $c$, we then arrive at the optimal consumption strategy $c^*$ defined as \eqref{cstar1}. This completes the proof.
\end{proof}

To obtain a more explicit form of the optimal policy, we search for solutions where $F^{(1)}$, $F^{(2)}$ and $F^{(3)}$  are tractable. We report the new derived forms below for the optimal solutions given in Lemma \ref{l33}, and we also verify the new solutions are well-defined.

\begin{theorem}
The optimal investment and consumption strategies for model \eqref{maxu2} are
\begin{equation}\label{pi2}
\pi^*(t)x=\frac{\mu-r}{\sigma^2\gamma f(t)}\left(a(t)+\gamma(a^2(t)-f(t))\right)(x+K^{(c^*)}(t))
\end{equation}
provided that
\begin{equation}\label{cond1}
\frac{\gamma f(t)}{x+K^{(c^*)}(t)}>0,
\end{equation}
and
\begin{equation}\label{cstar2}
c^*(t)=[U']^{-1}\left(a(t)+\frac{\gamma}{2}(a^2(t)-f(t))\right)
\end{equation}
respectively, and the optimal objective value is
\begin{equation}
F(t,x)=a(t)\left(x+K^{c^*}(t)\right)-\frac{\gamma}{2}\left(f(t)-a^2(t)\right)\left(x+K^{c^*}(t)\right)+\beta\int_{t}^{T}e^{-\rho s}U(c^*(s))ds.
\end{equation}
where $a(t)$ and $f(t)$ are given by
\begin{equation}
\frac{da(t)}{dt}=-\left((r-\delta)+\frac{\mu-r}{\sigma^2\gamma f(t)}(a(t)+\gamma(a^2(t)-f(t)))\right)a(t)
\end{equation}
and
\begin{equation}
\frac{df(t)}{dt}=-2\left((r-\delta)+\frac{\mu-r}{\sigma^2\gamma f(t)}\left(a(t)+\gamma(a^2(t)-f(t))\right)\right)+\frac{\mu-r}{\sigma^2\gamma f(t)}\left(a(t)+\gamma(a^2(t)-f(t))\right)f(t)
\end{equation}
with initial conditions $a(T)=f(T)=1$.
\end{theorem}
\begin{proof}
To obtain an explicit solution for this optimal portfolio selection problem, we assume that $F^{(1)}$, $F^{(2)}$ and $F^{(3)}$ can be written in the following form:
\begin{equation}\label{aa1}
\begin{aligned}
F^{(1)}(t,x)&=a(t)(x+K^{c^*}(t))+b(t), \\
F^{(2)}(t,x)&=f(t)(x+K^{c^*}(t))^2+g(t)(x+K^{c^*}(t))+h(t),\\
F^{(3)}(t,x)&=p(t)(x+K^{c^*}(t))+q(t),
\end{aligned}
\end{equation}
where $a$, $b$, $f$, $g$, $h$, $p$ and $q$ are deterministic functions of time. The candidate for the optimal consumption strategy $c$ is assumed to be independent of wealth, which implies that
\begin{equation}\label{KK}
K^{c^*}(t)=\int_{t}^{T}e^{-r(s-t)}(l(s)-c^*(s))ds.
\end{equation}
We also assume that
\begin{equation}\label{assume}
a(t)b(t)=\frac{g(t)}{2},\quad h(t)=b^2(t).
\end{equation}
Substituting \eqref{aa1} into \eqref{pi1} and \eqref{cstar1} can yield the new forms \eqref{pi2} and \eqref{cstar2}. Now insert \eqref{pi2} and \eqref{cstar2} into \eqref{Y1} and include the terminal conditions to get
\begin{equation}
\frac{da(t)}{dt}=-\left((r-\delta)+\frac{\mu-r}{\sigma^2\gamma f(t)}(a(t)+\gamma(a^2(t)-f(t)))\right)a(t)
\end{equation}
and
\begin{equation}
\frac{db(t)}{dt}=\delta b(t)
\end{equation}
with terminal conditions $a(T)=1$ and $b(T)=0$, respectively.
In the same way, substituting \eqref{pi2} and \eqref{cstar2} into \eqref{Z1} yields
\begin{equation}
\begin{aligned}
\frac{df(t)}{dt}&=-2\left((r-\delta)+\frac{\mu-r}{\sigma^2\gamma f(t)}(a(t)+\gamma(a^2(t)-f(t)))\right)a(t)+\frac{\mu-r}{\sigma^2\gamma f(t)}(a(t)+\gamma(a^2(t)-f(t)))f(t)\\
\frac{dg(t)}{dt}&=-\left(r+\frac{(\mu-r)^2}{\sigma^2\gamma f(t)}(a(t)+\gamma(a^2(t)-f(t)))\right)g(t)+2\rho g(t),
\end{aligned}
\end{equation}
and
\begin{equation}
\frac{dh(t)}{dt}=2\delta h(t),
\end{equation}
with terminal conditions $f(T)=1$, $g(T)=h(T)=0$. By inserting \eqref{pi2} and \eqref{cstar2} into \eqref{W1}, we then have
\begin{equation}
\frac{dp(t)}{dt}=-\left((r-\delta)+\frac{\mu-r}{\sigma^2\gamma f(t)}(a(t)+\gamma(a^2(t)-f(t)))\right)p(t)
\end{equation}
and
\begin{equation}
\frac{dq(t)}{dt}=e^{-\rho t} U(c^*(t))dt
\end{equation}
with terminal conditions $p(T)=q(T)=0$.

After simple calculations, we further have $b(t)=g(t)=h(t)=p(t)=0$, which guarantees the assumptions \eqref{assume}. Besides, $q(t)=\int_{t}^{T}e^{-\rho s}U(c^*(s))ds$.

In addition, the optimal value function $F$ can be also derived as
\begin{equation}
\begin{aligned}
F(t,x)&=F^{(1)}-\frac{\gamma}{2(x+K^{(c^*)}(t))}\left(F^{(2)}-(F^{(1)})^2\right)+\beta F^{(3)}\\
&=a(t)\left(x+K^{c^*}(t)\right)+b(t)-\frac{\gamma}{2}\left(f(t)-a^2(t)\right)\left(x+K^{c^*}(t)\right)+\beta\left(p(t)(x+K^{c^*}(t))+q(t)\right)\\
&=a(t)\left(x+K^{c^*}(t)\right)-\frac{\gamma}{2}\left(f(t)-a^2(t)\right)\left(x+K^{c^*}(t)\right)+\beta\int_{t}^{T}e^{-\rho s}U(c^*(s))ds.
\end{aligned}
\end{equation}
This completes the proof.
\end{proof}

\begin{remark}
It follows \eqref{wealth} and \eqref{KK} that
\begin{equation}
\begin{aligned}
d\left(X^{c^*,\pi^*}(t)+K^{(c^*)}(t)\right)=&\left(r+\frac{(\mu-r)^2}{\sigma^2\gamma f(t)}(a(t)+\gamma(a^2(t)-f(t)))\right)\left(X^{c^*,\pi^*}(t)+K^{(c^*)}(t)\right)dt\\
&+\frac{\mu-r}{\sigma^2\gamma f(t)}\left(a(t)+\gamma(a^2(t)-f(t)))\right)\left(X^{c^*,\pi^*}(t)+K^{(c^*)}(t)\right)dB(t).
\end{aligned}
\end{equation}
Therefore,
\begin{equation}\label{Xt}
\begin{aligned}
&X^{c^*,\pi^*}(t)+K^{(c^*)}(t)\\
=&(x_0+K^{(c^*)}(0))\exp\bigg[\int_0^t \left(r+\frac{(\mu-r)^2}{\sigma^2\gamma f(s)}(a(s)+\gamma(a^2(s)-f(s)))-\frac{1}{2}\frac{(\mu-r)^2}{\sigma^2\gamma f(s)}(a(s)+\gamma(a^2(s)-f(s)))^2\right)ds\\
&+\int_0^t\frac{\mu-r}{\sigma^2\gamma f(s)}\left(a(s)+\gamma(a^2(s)-f(s)))\right)dB(s)\bigg].
\end{aligned}
\end{equation}
Since the initial condition ensures $x_0+K^{(c^*)}(0)>0$ and $f$ is proved to be strictly positive in \eqref{e33} below, we conclude the condition \eqref{cond1} for the optimal investment strategy is fulfilled.
\end{remark}

\begin{remark}
The system  composed of PDEs \eqref{Y1}, \eqref{Z1} and \eqref{W1} has a unique global solution. In fact, by replacing the integral interval to $[t,T]$ and taking conditional expectation at $t$ in \eqref{Xt}, we have
\begin{equation}\label{e11}
E\left[\left.X^{(c^*,\pi^*)}(T)\right|X(t)=x\right]=(x+K^{c^*}(t))\exp\left(\int_{t}^{T}[r+(\mu-r)\tilde{\pi}^*(s)]ds\right)
\end{equation}
and
\begin{equation}\label{e12}
E\left[\left.(X^{(c^*,\pi^*)}(T))^2\right|X(t)=x\right]=(x+K^{c^*}(t))^2\exp\left(2\int_{t}^{T}[r+(\mu-r)\tilde{\pi}^*(s)+\frac{1}{2}\sigma^2(\tilde{\pi}^*(s))^2]ds\right),
\end{equation}
where
\begin{equation}\label{pistar}
\tilde{\pi}^*(s)=\frac{\mu-r}{\sigma^2\gamma f(t)}\left(a(t)+\gamma(a^2(t)-f(t))\right).
\end{equation}
Comparing \eqref{aa1} with \eqref{e11} and \eqref{e12} yields
\begin{equation}
\begin{aligned}
&a(t)(x+K^{c^*}(t))+b(t)
=E\left[\left.e^{-\delta(T-t)}X^{(c^*,\pi^*)}(T)\right|X(t)=x\right]\\
=&(x+K^{c^*}(t))\exp\left(\int_{t}^{T}[(r-\delta)+(\mu-r)\tilde{\pi}^*(s)]ds\right)
\end{aligned}
\end{equation}
and
\begin{equation}
\begin{aligned}
&f(t)(x+K^{c^*}(t))^2+g(t)(x+K^{c^*}(t))+h(t)
=E\left[\left.(e^{-\delta(T-t)}X^{(c^*,\pi^*)}(T))^2\right|X(t)=x\right]\\
=&(x+K^{c^*}(t))^2\exp\left(2\int_{t}^{T}[(r-\delta)+(\mu-r)\tilde{\pi}^*(s)+\frac{1}{2}\sigma^2(\tilde{\pi}^*(s))^2]ds\right).
\end{aligned}
\end{equation}
Collecting terms we obtain
\begin{equation}\label{e33}
\begin{aligned}
a(t)&=\exp\left(\int_{t}^{T}[(r-\delta)+(\mu-r)\tilde{\pi}^*(s)]ds\right),\\
f(t)&=\exp\left(2\int_{t}^{T}[(r-\delta)+(\mu-r)\tilde{\pi}^*(s)+\frac{1}{2}\sigma^2(\tilde{\pi}^*(s))^2]ds\right),
\end{aligned}
\end{equation}
and $b(t)=g(t)=h(t)=0$.
Substituting \eqref{e33} into \eqref{pistar} leads to
\begin{equation}\label{inter}
\tilde{\pi}^*(t)=\frac{\mu-r}{\sigma^2\gamma}\left(e^{-\int_{t}^{T}[(r-\delta)+(\mu-r)\tilde{\pi}^*(s)+\sigma^2(\tilde{\pi}^*(s))^2]ds}+\gamma e^{-\int_{t}^{T}\sigma^2(\tilde{\pi}^*(s))^2]ds}-\gamma\right).
\end{equation}
By designing the algorithm as $\tilde{\pi}_0(t)=1$
and
$$\tilde{\pi}_{n+1}(t)=\frac{\mu-r}{\sigma^2\gamma}\left(e^{-\int_{t}^{T}[(r-\delta)+(\mu-r)\tilde{\pi}_n(s)+\sigma^2(\tilde{\pi}_n(s))^2]ds}+\gamma e^{-\int_{t}^{T}\sigma^2(\tilde{\pi}_n(s))^2]ds}-\gamma\right)$$ for $n\geq 1$ on $[0,T]$,
we can prove that the sequence $\{\tilde{\pi}_n\}$ converges to the solution $\tilde{\pi}^*$, which verifies the uniqueness of the optimal investment strategy.
\end{remark}

\section{Concluding remarks}

In this paper, we introduced the time and state dependent risk aversion into the mean-variance-utility portfolio selection problem and a new portfolio selection model embraces is proposed. We solved the model under a game theoretic framework and analytically derived the continuous equilibrium investment (consumption) policy.  The results perform economically reasonable implication that optimal investment strategy heavily depends on the investor's current wealth and future income-consumption balance. In addition, the continuous optimally consumption process shows high dependence on the investor's consumption preference.

\end{document}